\documentclass[conference]{IEEEtran}
\IEEEoverridecommandlockouts
\usepackage{amsmath,amsfonts}
\usepackage{algorithmic}
\usepackage{array}
\usepackage[caption=false,font=normalsize,labelfont=sf,textfont=sf]{subfig}
\usepackage{textcomp}
\usepackage{multirow}
\usepackage{stfloats}
\usepackage{url}
\usepackage{verbatim}
\usepackage{graphicx}
\usepackage{amssymb}
\usepackage{amsthm}
\usepackage{hyperref}
\usepackage{tikz}
\usepackage{float}
\usepackage{subfig}
\usepackage{orcidlink}
\usepackage{cite}
\usepackage{caption}
\usepackage{mathrsfs}
\usepackage{cuted}
\usepackage{booktabs}
\usepackage{tabularx}   % 提供 tabularx 环境（自动等宽表格）
\usepackage{booktabs}   % 提供 \toprule、\midrule、\bottomrule 等横线命令

\usetikzlibrary{arrows.meta, positioning}
\def\BibTeX{{\rm B\kern-.05em{\sc i\kern-.025em b}\kern-.08em
		T\kern-.1667em\lower.7ex\hbox{E}\kern-.125emX}}
\usepackage{balance}
\hypersetup{
	colorlinks=true,
	linkcolor=blue,
	citecolor=blue
}
\newtheoremstyle{mystyle}% Name
{3pt}% Space above
{3pt}% Space below
{}% Body font
{1em}% Indent amount (approximately 1em for IEEE style)
{\bfseries}% Theorem head font
{}% Punctuation after theorem head (none)
{.5em}% Space after theorem head
{}% Theorem head spec

\theoremstyle{mystyle}
\newtheorem{assumptionx}{Assumption}
\newtheorem{lemmax}{Lemma}
\newtheorem{definitionx}{Definition}
\newtheorem{theoremx}{Theorem}
\newtheorem{remarkx}{Remark}
\newenvironment{assumption}
{\begin{assumptionx}\hspace{0.5em}}
	{\end{assumptionx}}

\newenvironment{lemma}
{\begin{lemmax}}
	{\end{lemmax}}

\newenvironment{theorem}
{\begin{theoremx}\hspace{-0.5em}\textnormal{:}}
	{\end{theoremx}}
\newenvironment{remark}
{\begin{remarkx}\hspace{-0.5em}\textnormal{:}}
	{\end{remarkx}}
\begin{document}
\title{Singularity-free prescribed performance guaranteed
 control for perturbed system}

\author{
    \IEEEauthorblockN{1\textsuperscript{st} Yiwei Liu}
    \IEEEauthorblockA{
        \textit{Chongqing Key Laboratory of Nonlinear Circuits and Intelligent Information Processing} \\
        \textit{College of Westa, Southwest University} \\
        Chongqing 400715, China \\
        ORCID: 0009-0008-1283-8245
    }
}

	\maketitle
	
\begin{abstract}
This paper addresses the prescribed performance control (PPC) challenge for high-order nonlinear systems affected by mismatched disturbances. The research aims to prevent singularity issues arising from error boundary violations during abrupt changes in reference trajectories. We introduce a novel transformation function with infinite-order differentiability at connection points, advancing beyond mere continuous differentiability. Utilizing this transformation function, we develop a comprehensive transformation strategy that ensures: (1) errors remain within prescribed boundaries when reference trajectories are smooth, and (2) errors return to prescribed boundaries within a specified timeframe following abrupt changes in reference trajectories. Additionally, the complexity explosion issue inherent in backstepping design is effectively resolved. Simulation results corroborate the validity of the proposed theoretical advancements.
\end{abstract}


\section{Introduction}

For decades, nonlinear control has garnered increasing attention due to its extensive applications in real-world systems, including boost converters \cite{10274473,liu2025fixedtimevoltageregulationboost,11048942} and autonomous vehicles \cite{wang2025adaptive,10878311}. However, controlling nonlinear systems often poses significant challenges, particularly when external disturbances are present. These challenges necessitate the development of more robust and adaptive control strategies to ensure system stability and performance. As research in this field continues to evolve, addressing these complexities will be crucial for advancing the capabilities and reliability of modern technological systems.

The strict-feedback form represents the most prevalent class of nonlinear systems. Notably, many nonlinear systems with varying structures can be transformed into strict-feedback systems through diffeomorphic transformations, simplifying the process of controller design. The primary challenge in this context is managing lumped error terms, which typically comprise both external disturbances and system uncertainties, such as actuator faults and hysteresis effects. Current approaches to disturbance handling include adaptive parameter estimation \cite{10233088,LIU2025107422}, disturbance observers \cite{9900364,wang2024hybrid}, and approximations using neural networks (NNs) \cite{9343685,wang2024event} or fuzzy logic systems (FLS) \cite{9954896}. Despite their effectiveness, these methods necessitate additional compensation structures, which can significantly increase system complexity. This complexity underscores the need for innovative solutions that can efficiently handle disturbances while maintaining system simplicity and performance.

Building upon these disturbance handling methods, backstepping can be utilized for controller design. However, when applied to high-order nonlinear systems, traditional backstepping encounters the "complexity explosion" problem, wherein the virtual control laws for higher-order subsystems necessitate derivatives of lower-order subsystems, resulting in rapidly escalating computational loads and symbolic expression complexity. Current solutions primarily consist of: (1) dynamic surface control (DSC) \cite{880994,wangCC} and command filters \cite{10637464}, which employ low-pass filters for derivative approximation; and (2) NNs \cite{gong2025effective,9968116,icca} or FLS \cite{9954896}, which directly approximate complex terms containing virtual controller derivatives. Both approaches introduce increased structural complexity and potential approximation errors that can compromise system stability. Furthermore, NNs and FLS implementations require meticulously designed training strategies to ensure convergence of weights to optimal values, adding another layer of complexity to the control design process. These challenges highlight the need for more efficient methods that maintain stability without excessively complicating the system architecture.

To enhance control performance, backstepping is frequently integrated with prescribed performance control (PPC), which is fundamentally designed to constrain system errors within predefined boundaries, thereby ensuring stable operation \cite{LIU2025107422,Zhang07012025,SUI2025120005}. Existing PPC approaches can be classified into two main categories. The first category involves initial-condition-dependent PPC functions, characterized by simpler transformation laws that necessitate initial errors to be within specified boundaries. This approach offers reduced computational complexity and is effective in situations where initial conditions can be precisely controlled \cite{Zhang07012025,SUI2025120005}. The second category comprises initial-condition-independent PPC functions, which, although inherently more complex, provide greater flexibility by relaxing initial condition requirements. This makes them particularly suitable for applications where initial states are uncertain or dynamic, allowing for robust control despite unpredictable initial conditions \cite{LIU2025107422}. The choice between these strategies depends on the specific needs and constraints of the application, such as computational resources and the predictability of initial conditions, thereby facilitating enhanced robustness and performance in various operational environments.

Notably, PPC serves as a common method for achieving predefined-time stability. By designing performance boundaries as piecewise functions that become constant after a predefined time, system stability can be guaranteed if errors remain within these boundaries. Compared with alternative predefined-time control methods, PPC avoids singularity issues that arise from fractional exponents in traditional approaches, where differentiation produces negative exponents causing instability as errors approach zero \cite{9954896,10529316,10058598}. However, PPC faces a critical challenge: when reference trajectories contain discontinuities, error jumps may violate performance boundaries, causing singularity problems. This scenario frequently occurs in practical applications like boost converters where reference outputs need step changes between different voltage levels.

Current research on PPC singularity primarily addresses boundary violations due to state measurement errors, while relatively little attention has been given to the challenges posed by reference trajectory jumps \cite{GONG2023251,9632614}. To bridge this gap, we propose a singularity-resistant PPC method that introduces three key contributions:

\begin{itemize}
\item First, we design a novel globally smooth transformation function that strictly avoids potential singularity issues caused by nonexistent high-order derivatives.
\item Second, based on the proposed transformation function, we develop a new control strategy that achieves prescribed performance. This guarantees that when the reference trajectory undergoes abrupt changes, the error will return to the prescribed boundary within a predetermined time, while maintaining the error within the boundary when the reference trajectory is smooth.
\item Third, the proposed control scheme features low complexity, specifically manifested in its inherent ability to reject external disturbances without requiring additional strategies, while simultaneously resolving the complexity explosion problem in backstepping design.
\end{itemize}

The structure of this paper is organized as follows: Section~\ref{sec2} formulates the problem and establishes the control objectives. Section~\ref{sec:main_results} develops the control strategy based on backstepping methodology. Section~\ref{sec:stability} provides rigorous stability analysis of the proposed system. Section~\ref{sec5} demonstrates the method's effectiveness through numerical simulations. Section~\ref{sec6} concludes the paper and suggests future research directions.

\textbf{Notations}
$d_S(a, B) = \inf_{b \in B} \|a - b\|$ denotes the smallest distance from the point $a$ to the set $B$.
\section{Problem Formulation}\label{sec2}
First, we consider a higher-order strict-feedback system with positional disturbance terms:

\begin{equation}
\begin{cases}
\dot{x}_{i} = x_{i+1} + d_{i}, & i = 1,2,3,\cdots,n, \\ 
y = x_{1}
\end{cases}
\label{eq:system}
\end{equation}

where $x_{i} \in \mathbb{R}$ represents the system state, $d_{i}$ denotes the mismatched disturbance, and $y$ is the system output. Specifically, when $i = n$, $x_{n+1} = u$, where $u$ is the control input.

\begin{assumption}
The disturbance terms $d_{i}$ are bounded, like \cite{10574398}.
\end{assumption}

\begin{assumption}
The reference trajectory $x_{r}$ is piecewise smooth and may contain finite jump points.
\end{assumption}

Define the tracking error $e_{1} = x_{1} - x_{r}$. The control goal naturally transforms to designing $u$ that drives $e_{1}$ to converge to zero. Following the backstepping principle, we further define virtual errors:

\begin{equation}
e_{i+1} = x_{i+1} - \alpha_{i}, \quad i = 1,2,3,\cdots,n-1
\label{eq:virtual_errors}
\end{equation}

Here, $\alpha_i$ is the virtual control law, and its specific form will be given later in the text. To ensure notational consistency, we define $x_{r} = \alpha_{0}$. This provides a unified definition for both tracking errors and virtual errors:

\begin{equation}
e_{i} = x_{i} - \alpha_{i-1}, \quad i = 1,2,3,\cdots,n.
\label{eq:unified_errors}
\end{equation}

To enhance control performance for higher-order states, we define error boundaries $\rho$. The overall control objective can thus be stated as:

Design the control input $u$ such that each error satisfies the prescribed performance requirements.

\subsection{Singularity Problem Analysis}
To address the singularity problem arising from reference trajectories containing finite jump discontinuities, let us consider a simple step reference trajectory defined as:

\begin{equation}
x_{r}(t) = 
\begin{cases}
C_{1}, & t \in (0,t_{1}] \\
C_{2}, & t \in (t_{1},t_{2}] \\
C_{3}, & t \in (t_{2},t_{3}] \\
\vdots
\end{cases}
\label{eq:step_reference}
\end{equation}

where $C_{1},C_{2},C_{3},\ldots$ are distinct constants representing different reference levels. When the error $e_1$ is successfully constrained within $(-\rho,\rho)$, any sudden jump in $x_{r}$ will inevitably cause a corresponding jump in $e_1$, making it exceed the error boundary and thus creating the aforementioned singularity problem.

\subsection{Control Objectives}
We now formally state the control objectives for system (\ref{eq:system}):

\begin{itemize}
\item For smooth reference trajectories (without jump discontinuities), ensure the tracking error remains strictly within prescribed boundaries continuously.
\item For discontinuous reference trajectories (with jump discontinuities), guarantee that:
\begin{enumerate}
\item The tracking error returns to prescribed boundaries within finite time after each jump.
\item The error remains bounded until subsequent reference jumps occur.
\end{enumerate}
\end{itemize}

	\section{Main Results}\label{sec:main_results}
    	\begin{figure}[!t]
		\centering
		\includegraphics[width=1\linewidth]{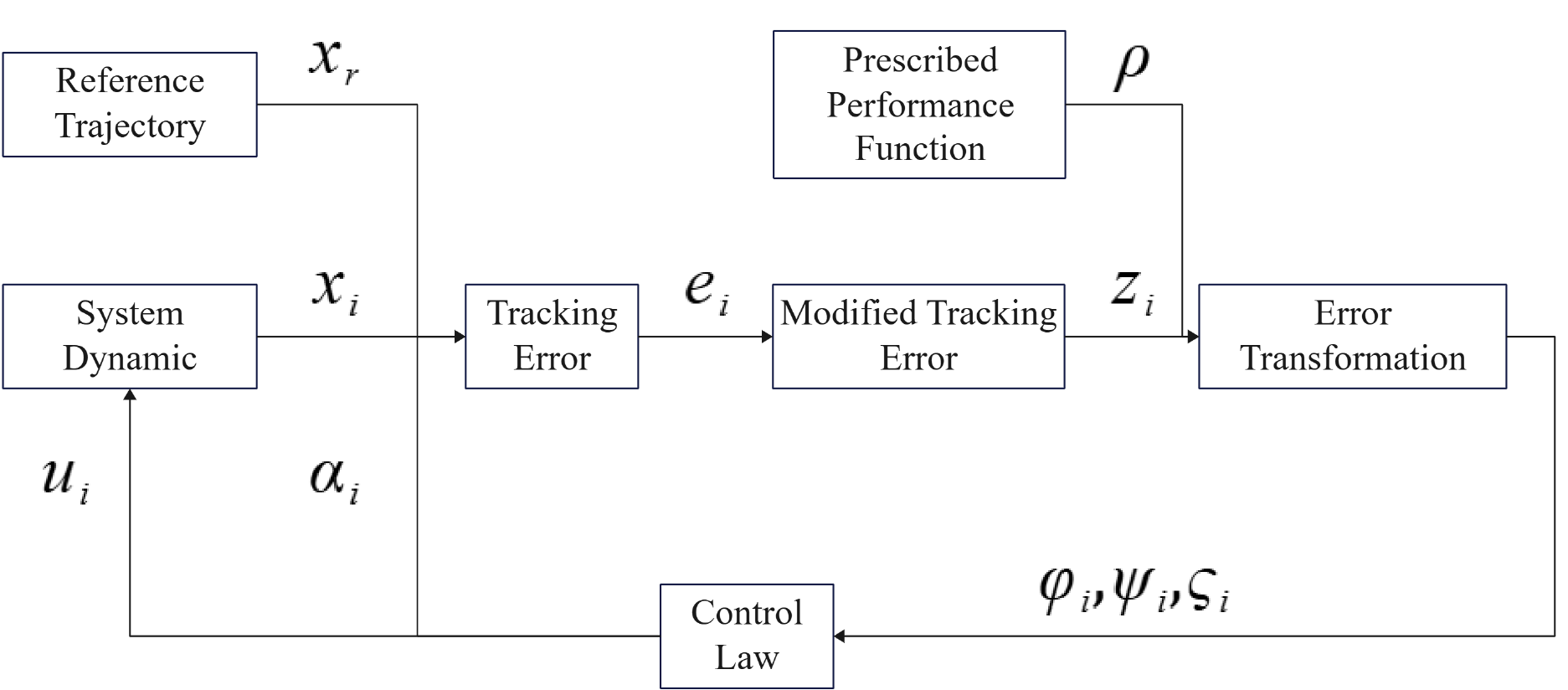}
		\caption{The overall framework.}
		\label{framework}
	\end{figure}
In this section, we will first define a transformation function to rectify the error definition, thereby addressing the discontinuity issue in the reference trajectory. Subsequently, based on the rectified error, we will design an error transformation and develop a controller accordingly. The overall framework is illustrated in Fig. \ref{framework}.

\subsection{Transformation Function Design}
To address the singularity issue, we first define a transformation function:

\begin{equation}
\mu(t) = 
\begin{cases}
0 & \text{if } t \leq 0, \\ 
\frac{1}{C}\int_{0}^{t}\exp\left(-\frac{1}{s(T-s)}\right)\mathrm{d}s & \text{if } 0 < t < T, \\ 
1 & \text{if } t \geq T,
\end{cases}
\label{eq:mu_function}
\end{equation}

where $C = \int_{0}^{T}\exp\left(-\frac{1}{s(T-s)}\right)\mathrm{d}s$ is the normalization constant. Thus, it can be readily shown that $\mu$ is strictly monotonically increasing on the interval $[0,T]$. Next, we prove its smoothness at $t=T$.

\begin{lemma}
The function $\mu(t)$ is smooth at $t=T$, meaning it is infinitely differentiable at this point.
\end{lemma}

\begin{proof}
For $0<t<T$, the $n$-th derivative of $\mu(t)$ satisfies:

\begin{equation}
\mu^{(n)}(t) = \frac{1}{C} \cdot \frac{\mathrm{d}^{n-1}}{\mathrm{d}t^{n-1}} \exp\left(-\frac{1}{t(T-t)}\right).
\label{eq:nth_derivative}
\end{equation}

The kernel function $\exp\left(-\frac{1}{t(T-t)}\right)$ is flat at $t=T$, i.e., all its derivatives vanish as $t\to T^{-}$. 

First, we prove continuity:

\begin{equation}
\lim_{t\to T^{-}}\mu(t) = \frac{\int_{0}^{T}\exp\left(-\frac{1}{s(T-s)}\right)\mathrm{d}s}{C} = 1.
\label{eq:continuity}
\end{equation}

Thus, $\mu$ is continuous at $t=T$.

Next, we prove first-order differentiability:

\begin{equation}
\lim_{t\to T^{-}}\mu'(t) = \lim_{t\to T^{-}}\frac{1}{C}\exp\left(-\frac{1}{t(T-t)}\right) = 0.
\label{eq:first_derivative}
\end{equation}

For higher-order derivatives ($n \geq 2$), we observe that:

\begin{equation}
\lim_{t\to T^{-}}\mu^{(n)}(t) = \lim_{t\to T^{-}}\frac{P_n(t)}{[t(T-t)]^{k_n}}\exp\left(-\frac{1}{t(T-t)}\right) = 0,
\label{eq:higher_derivatives}
\end{equation}

where $P_n(t)$ is a polynomial in $t$ and $t_n$ is a positive integer depending on $n$ . Therefore, all derivatives of $\mu$ exist and are continuous at $t=T$, which completes the proof.
\end{proof}

\begin{remark}
Compared with existing shift functions \cite{10574398,10375267}, most current implementations only guarantee $C^{1}$ or $C^{2}$ continuity. For high-order systems requiring higher derivatives of the shift function, the nonexistence of these derivatives at discontinuity points may generate impulse signals that could compromise system stability. The proposed $\mu$ function, being globally smooth ($C^{\infty}$), fundamentally avoids this issue.
\end{remark}
\subsection{Modified Tracking Error Definition}
When the reference trajectory undergoes abrupt changes, we explicitly know the jump time instants $t_i$, $i = 1, 2, 3, \cdots$. To avoid ambiguity, we specify that $t_i = +\infty$ when the reference trajectory $x_r$ contains no jump discontinuities. We then shift $\mu$ rightward along the time axis by $t_i$, obtaining $\mu(t - t_i)$. For notational simplicity, we define $\mu_i = \mu(t - t_i)$.

The modified tracking error can thus be defined as:

\begin{equation}
z_i = 
\begin{cases} 
e_i, & \text{if } t \in [0, t_1], \\
\mu_1 e_i, & \text{if } t \in (t_1, t_2], \\
\mu_2 e_i, & \text{if } t \in (t_2, t_3], \\
\vdots &
\end{cases}
\label{eq:modified_error}
\end{equation}

This formulation shows that $z_i = e_i$ when $x_r$ contains no jump discontinuities. Notably, since $t_i$ represents the jump instant, we can select $T$ such that $T < t_{i+1} - t_i$.

\begin{remark}
The multiplication of each $e_i$ by $\mu_i$ is crucial because reference jumps induce discontinuity propagation through all error channels via backstepping. Jumps in $x_r$ directly affect $e_1$, which modifies $\alpha_1$ and consequently alters $e_2$, with this effect propagating through all higher-order errors $e_i$. Therefore, the $\mu_i$ scaling must be applied at every error level.
\end{remark}

\subsection{Control Strategy Design}
We define the error boundary and coordinate transformation functions\cite{LIU2025107422}:

\begin{equation}
\left\{
\begin{aligned} 
\rho &= \begin{cases}
\tan\left(\frac{\pi}{2h}\right), & t \in (0,T_{1}], \\
c, & t \in (T_{1},+\infty),
\end{cases} \\
\frac{1}{h} &= \begin{cases}
l + (1-l)\cos^{2}\left(\frac{\pi t}{2T_{1}}\right), & t \in (0,T_{1}], \\
l, & t \in (T_{1},+\infty),
\end{cases} \\
\varphi_{i} &= \frac{2}{\pi}\arctan(z_{i}), \\
\psi_{i} &= \varphi_{i}h, \\
\varsigma_{i} &= \operatorname{artanh}(\psi_{i}).
\end{aligned}
\right.
\label{eq:transformation}
\end{equation}

where $T_{1}$ denotes the predefined time, $c$ represents the preset range, and $l = \frac{2}{\pi}\arctan(c)$. Following Lemma 2 in prior work, prescribed performance can be guaranteed by ensuring the boundedness of $\varsigma_{i}$.

The backstepping control strategy is formulated as:

\begin{equation}
\alpha_{i} = -k_{i}\varrho_{i}\varsigma_{i}
\label{eq:control_law}
\end{equation}

where $\varrho_{i} = \frac{1}{1-\psi_{i}^{2}}$, and $k_{i}$ is a designable positive constant. Specifically, $\alpha_{n} = u$ serves as the actual control input.

\begin{remark}
Unlike \cite{LIU2025107422}, our coordinate transformation guarantees the boundedness of $\varsigma_{i}$, which ensures $z_{i} \in (-\rho,\rho)$ rather than directly guaranteeing $e_{i} \in (-\rho,\rho)$. When the reference trajectory contains no jump discontinuities, this is completely equivalent to the direct control of $e_i$ in \cite{LIU2025107422}. However, when jump discontinuities exist in the reference trajectory, $z_i$ will always remain within the prescribed performance bounds. According to the transformation relationship between $z_i$ and $e_i$, it can be concluded that $e_i$ will return within the prescribed bounds within time $T$, thereby ensuring the recovery of PPC effectiveness within the prescribed time frame.
\end{remark}

\begin{remark}
Analogous to \cite{LIU2025107422}, we extend the results in \cite{LIU2025107422} to higher-order disturbed systems. Our control law maintains low complexity by excluding: (1) derivatives of virtual control laws $\alpha_{i}$, and (2) derivatives of $x_{r}$. This design resolves the complexity explosion problem and relaxes assumptions on $x_{r}$ by eliminating sensitivity to nondifferentiable points during reference jumps.
\end{remark}

\section{Stability Analysis}
\label{sec:stability}

\begin{theorem}
For the system (\ref{eq:system}), the proposed error transformation laws (\ref{eq:modified_error}), (\ref{eq:transformation}) and the control scheme (\ref{eq:control_law}) guarantee that:
\begin{enumerate}
\item When the reference signal contains no jump discontinuities, the error remains strictly within $(-\rho,\rho)$ at all times.
\item When the reference signal contains finite jump discontinuities, the error will ultimately return to $(-\rho,\rho)$.
\end{enumerate}
\end{theorem}

\begin{proof}
Consider the Lyapunov function candidate:
\begin{equation}
V_{i} = \frac{1}{2}\varsigma_{i}^{2}
\label{eq:lyapunov}
\end{equation}

Taking its time derivative yields:
\begin{equation}
\dot{V}_{i} = \varsigma_{i}\varrho_{i}\left(-\frac{2}{\pi}\frac{1}{1+z_{i}^{2}}\mu_{i}hk_{i}\varrho_{i}\varsigma_{i}\right) + \varsigma_{i}\varrho_{i}H_{i}
\label{eq:vdot}
\end{equation}

where
\begin{equation}
H_{i} = \frac{2}{\pi}\frac{1}{1+z_{i}^{2}}\dot{\mu}_{i}e_{i}h + \frac{2}{\pi}\frac{1}{1+z_{i}^{2}}\mu_{i}h(e_{i+1}+d_{i}-\dot{\alpha}_{i-1}+r_{i}\dot{h}).
\label{eq:Hi}
\end{equation}

Notably, while $H_{i}$ becomes discontinuous at jump points in $x_{r}$, it remains piecewise continuous since $x_{r}$ contains only finite jump points. Applying the Weierstrass theorem, we obtain:
\begin{equation}
|H_{i}| < \sigma_{i}
\label{eq:Hi_bound}
\end{equation}
where $\sigma_{i}$ is an unknown positive constant.

Since all terms in $-\frac{2}{\pi}\frac{1}{1+z_{i}^{2}}\mu_{i}hk_{i}$ are known, we can establish:
\begin{equation}
\varsigma_{i}\varrho_{i}\left(\frac{2}{\pi}\frac{1}{1+z_{i}^{2}}\mu_{i}hk_{i}\varrho_{i}\varsigma_{i}\right) \geq K_{i}(\varsigma_{i}\varrho_{i})^{2}
\label{eq:main_term}
\end{equation}
with $K_{i}$ satisfying $K_{i} < \frac{2}{\pi}\frac{1}{1+z_{i}^{2}}\mu_{i}hk_{i}$.

The derivative expression thus simplifies to:
\begin{equation}
\dot{V}_{i} \leq -K_{i}(\varsigma_{i}\varrho_{i})^{2} + \sigma_{i}\varsigma_{i}\varrho_{i}
\label{eq:vdot_simplified}
\end{equation}

Clearly, when $|\varsigma_{i}\varrho_{i}| > \frac{\sigma_{i}}{K_{i}}$, $\dot{V}_{i} < 0$. This proves the boundedness of both $\varsigma_{i}$ and $\varrho_{i}$. The boundedness of $\varsigma_{i}$ directly implies $z_{i}\in(-\rho,\rho)$.

When $x_{r}$ contains no jump discontinuities, $z_{i}\equiv e_{i}$, immediately proving $e_{i}\in(-\rho,\rho)$. When $x_{r}$ contains jump discontinuities, recalling the definition of $\mu_{i}$, we observe that after time $T$ following each jump, $z_{i}=e_{i}$ until the next reference signal jump occurs. This guarantees that $e_{i}$ will return to the prescribed boundary within $T$ after each reference signal jump.

Extending this analysis to the limit as $t\rightarrow+\infty$, we observe that
\begin{equation}
\left\|\psi_{i}+\frac{1}{d_{S}(\psi_{i},t)+\partial S\psi_{i}}\right\| < +\infty
\label{eq:asymptotic}
\end{equation}
holds even in the asymptotic case. According to Theorem 2 in \cite{VERGINIS2019538}, it can be  demonstrated that $\varsigma_{i}$ remains bounded as $t\rightarrow+\infty$. Consequently, we conclude that $z_{i}\in(-\rho,\rho)$ for all time, which completes the proof.
\end{proof}

	\section{Simulation Verification}\label{sec5}

To validate the correctness and effectiveness of the proposed control scheme, we consider the following second-order system:

\begin{equation}
\begin{cases}
\dot{x}_{1} = x_{2} + d_{1} \\ 
\dot{x}_{2} = u + d_{2}
\end{cases}
\end{equation}

For the error boundary function, we set $T_{1} = 1$ and $c = 0.1$. For $\mu$, we set $T = 2$. The control gains are chosen as $k_{i} = k_{2} = 50$. The reference trajectory is designed as:

\begin{equation}
x_{r} = 
\begin{cases}
\sin(t), & \text{if } t \in [0,3) \\ 
\sin(t) + 0.5, & \text{if } t \in [3,6) \\ 
\sin(t) - 0.5, & \text{if } t \in [6,10]
\end{cases}
\end{equation}

	\begin{figure}[!t]
		\centering
		\includegraphics[width=1\linewidth]{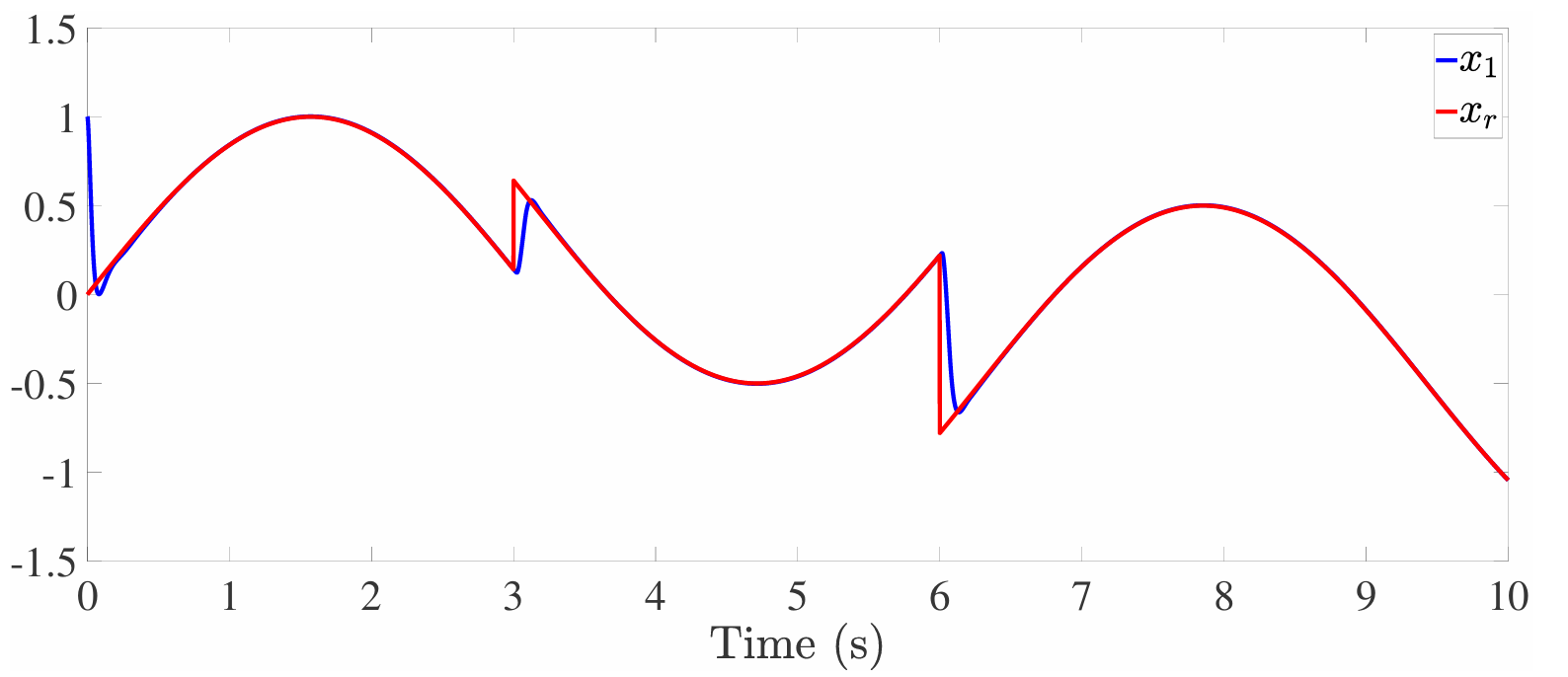}
		\caption{The first order state $x_1$.}
		\label{fig1}
	\end{figure}
    	\begin{figure}[!t]
		\centering
		\includegraphics[width=1\linewidth]{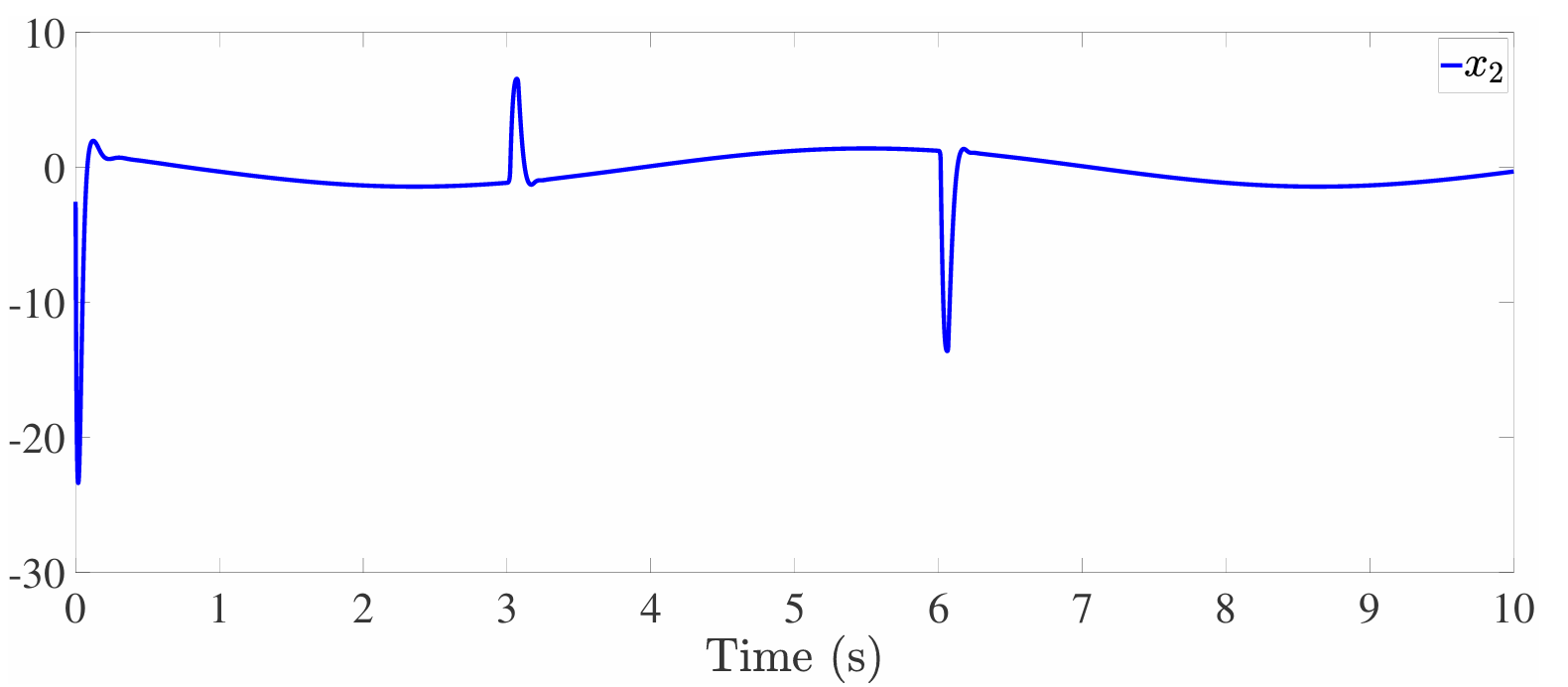}
		\caption{The second order state $x_2$.}
		\label{fig2}
	\end{figure}
    	\begin{figure}[!t]
		\centering
		\includegraphics[width=1\linewidth]{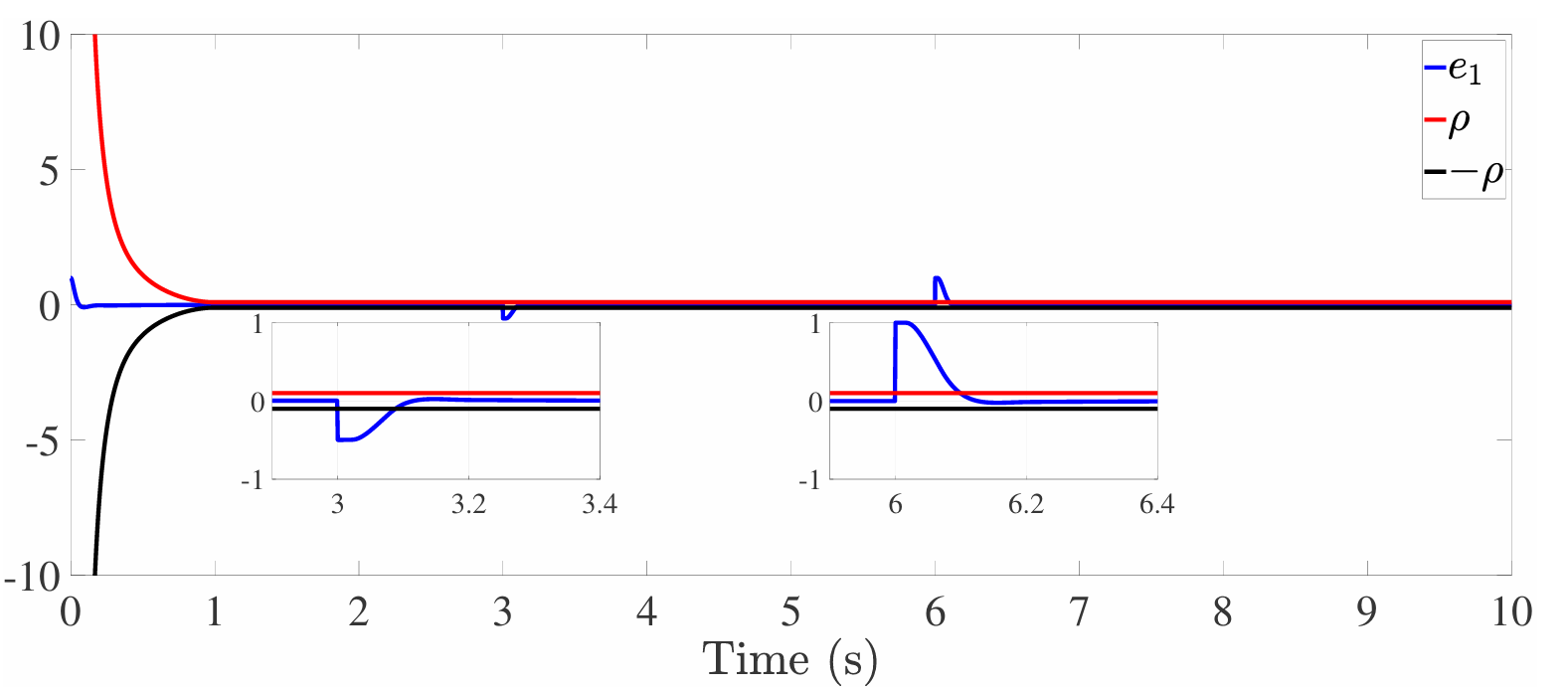}
		\caption{The tracking error $e_1$}
		\label{fig3}
	\end{figure}
    	\begin{figure}[!t]
		\centering
		\includegraphics[width=1\linewidth]{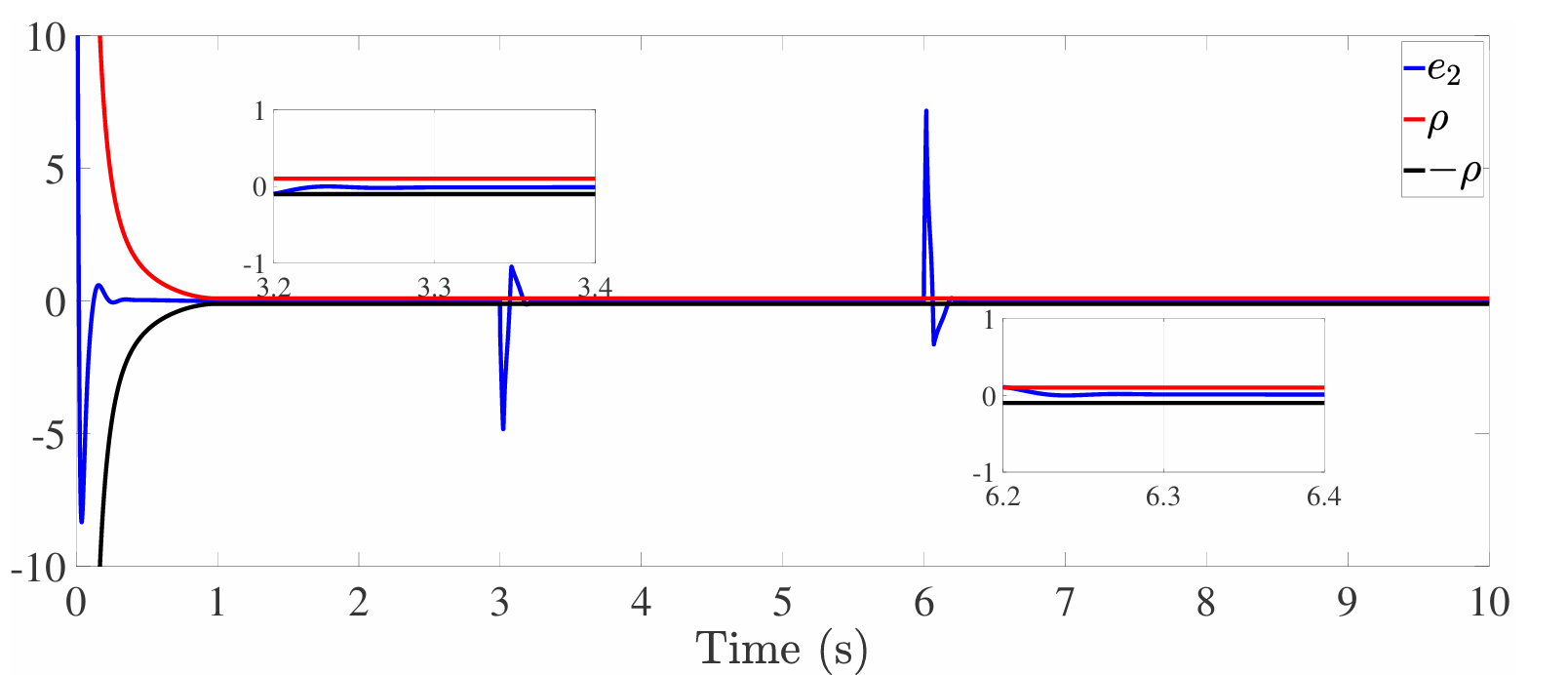}
		\caption{The virtual error $e_2$}
		\label{fig4}
	\end{figure}
    
The external disturbances are given by $d_{1} = \sin(t)$ and $d_{2} = 2\sin(t)$. The simulation results are shown in Figs. \ref{fig1}-\ref{fig4}. First, Fig. \ref{fig1} demonstrates successful achievement of the control objective - even when the reference trajectory undergoes abrupt changes, $x_{1}$ (the system output) can still rapidly track the reference trajectory.

Fig. \ref{fig2} reveals that $x_{2}$ exhibits jumps coinciding with reference trajectory discontinuities, which results from the control input mechanism. When reference jumps occur, the controller naturally generates large $u$ values to drive the system output back to $x_{r}$. Specifically, the jump in $e_{1}$ causes a corresponding jump in $\alpha$, which leads to a jump in $e_{2}$, consequently producing a jump in $u$ to ensure fast tracking.

Fig. \ref{fig4} validates this analysis. From Figs. \ref{fig3} and \ref{fig4}, we observe that errors remain within prescribed boundaries during steady-state periods, and even after reference jumps, both $e_{1}$ and $e_{2}$ return to their boundaries. These results conclusively prove the correctness and effectiveness of the proposed error transformation law.

\section{Conclusion}\label{sec6}

This paper examines the control problem for high-order strict-feedback systems subject to disturbances. We introduce a globally smooth transformation function designed to mitigate issues arising from nonexistent high-order derivatives. Building on this function, we propose an innovative error transformation strategy that ensures: (1) error recovery within a specified timeframe following abrupt changes in reference trajectories, and (2) the prevention of singularity issues during error transients. The proposed control strategy is characterized by its low complexity, effectively addressing external disturbances and resolving the complexity explosion problem without the need for additional mechanisms. Looking ahead, future research will focus on addressing singularity issues that may arise from error jumps during cyber attacks, further enhancing the robustness and resilience of the control system.

	\bibliographystyle{ieeetr}
	\bibliography{Reference.bib}
\end{document}